\documentclass[12pt,a4paper]{article}
\usepackage{fullpage}
\usepackage{amssymb,amsmath,stmaryrd,amsthm}
\usepackage[mathscr]{eucal}
\usepackage{pstricks,pst-node,pst-text,pst-3d}

\theoremstyle{definition}

\newtheorem{theorem}{Theorem}
\newtheorem{proposition}{Proposition}
\newtheorem{lemma}{Lemma}
\newtheorem{corollary}{Corollary}

\theoremstyle{remark}
\newtheorem{remark}{Remark}

\def \bR {\mathbb{R}}

\def \1{{\mathbf{1}}}
\def \cA {\mathcal{A}}

\def \cC {\mathcal{C}}
\def \cF {\mathcal{F}}
\def \cG {\mathcal{G}}

\def \cK {\mathcal{K}}

\def \cMG {\mathcal{MG}}

\def \cQ {\mathcal{Q}}
\def \cP {\mathcal{P}}

\def \gS {\mathfrak{S}}
\def \core {\mathsf{C}}
\def \mcore {\mathsf{MC}}
\def \conv {\mathrm{conv}}

\def \ext {\mathrm{ext}}
\def \PI {\mathsf{PI}}
\def \I {\mathsf{I}}
\def \sh {\mathsf{sh}}
\def \S {\mathsf{S}}
\def \W {\mathsf{W}}

\begin{document}

\title{On the set of imputations induced\\ by the $k$-additive core}

\author{Michel GRABISCH$^{1}$\thanks{Corresponding
    author. Tel (+33) 1-44-07-82-85, Fax
    (+33) 1-44-07-83-01,
    email \texttt{michel.grabisch@univ-paris1.fr}} and Tong LI$^{2}$ \\
\normalsize          1. Paris School of Economics, University of Paris I\\
\normalsize          106-112, Bd de l'H\^opital, 75013 Paris, France\\
\normalsize          2. Civil Aviation Management Institute of China\\
\normalsize           Beijing Institute of Technology\\
\normalsize      \#3,East Road Huajiadi, Chaoyang District, Beijing, China 100102  \\
\normalsize         Email: \texttt{michel.grabisch@univ-paris1.fr, ttlitong@gmail.com}}

\date{}
\maketitle

\begin{abstract}
An extension to the classical notion of core is the notion of $k$-additive core,
that is, the set of $k$-additive games which dominate a given game, where a
$k$-additive game has its M\"obius transform (or Harsanyi dividends) vanishing
for subsets of more than $k$ elements. Therefore, the 1-additive core coincides
with the classical core. The advantages of the $k$-additive core is that it is
never empty once $k\geq 2$, and that it preserves the idea of coalitional
rationality. However, it produces $k$-imputations, that is, imputations on
individuals and coalitions of at most $k$ inidividuals, instead of a classical
imputation. Therefore one needs to derive a classical imputation from a
$k$-order imputation by a so-called sharing rule. The paper investigates what
set of imputations the $k$-additive core can produce from a given sharing rule.
\end{abstract}


\section{Introduction}
A central problem in cooperative game theory is to define a rational way to
share the total worth of a game. Specifically, let $N$ be the set of players,
and $v$ a game in characteristic function form, assigning to each coalition
$S\subseteq N$ a worth $v(S)$, which in the case of profit game, represents the
benefit arising from the cooperation among members of $S$. Suppose that the best
way to generate profit is to form the grand coalition $N$. An important question
is how to share the total benefit $v(N)$ among the players. Any systematic way
of sharing $v(N)$ is called a solution of the game. 

The core \cite{gil53,sha71} is one of the most popular concept of solution. It
is defined as the set of preimputations which are coalitionally rational, i.e.,
there is no coalition $S$ such that the value $v(S)$ that $S$ can achieve by
itself is strictly greater than the payoff $x(S)$ given to $S$. This rationality
condition ensures that no coalition has interest to leave the grand coalition
$N$.

The main drawback of the core is that it is often empty, so that other concepts
of solution have to be sought for. The literature abunds on this topic, and many
new solution concepts have been proposed, for example the kernel \cite{dama65},
the selectope \cite{hapeso77,dehape00}, the nucleolus \cite{sch69}, the Shapley
value \cite{sha53} and so on.

Although all these propositions have their own merits, they depart from the
fundamental idea of coalitional rationality of the core. To keep as much as
possible this idea, Grabisch and Miranda have proposed the notion of
$k$-additive core \cite{grmi07,migr05}. Roughly speaking, the condition of
coalitional rationality $x(S)\geq v(S)$ for all $S\subseteq N$ is preserved, but
the notion of imputation/payoff is enlarged: $x(S)$ is no more the sum of
payoffs to individuals in $S$, i.e., $x(S)=\sum_{i\in S}x_i$, but it is a sum of payoffs
to individuals and possibly to coalitions of size at most $k$ in $S$. Such general
imputations are called $k$-order imputations. It is proved in
\cite{migr05} that as soon as $k=2$, the $k$-additive core is never empty. The
drawback is that eventually each player should receive an individual
payoff. Therefore, once a $k$-order imputation has been selected, it remains in
a second step to compute from it a classical imputation.

The aim of this paper is precisely to study what kind of imputation one can find
through the $k$-additive core. We will show that this question is closely
related to the selectope, and that surprisingly, any preimputation can be
attained through the 2-additive core of a game.

The paper is organized as follows. Section~\ref{sec:def} introduces the basic
material on the $k$-additive core and related notions. Section~\ref{sec:bare}
gives some basic results about the convex polytopes of the monotonic
$k$-additive core and the convex part of the $k$-additive core. Then,
Sections~\ref{sec:kadd} and \ref{sec:kadd1} give the main results of the paper,
i.e., the set of imputations induced by some classes of sharing values on the
$k$-additive core and the monotonic $k$-additive core.

Throughout the paper, we will often omit braces for singletons and sets. Also we
will write sets in capital italic, collections of sets in capital
calligraphic, and mappings either in small italic or sans serif. For any
vector $x\in\bR^n$ and $S\subseteq\{1,\ldots,n\}$, we use the shorthand
$x(S):=\sum_{i\in S}x_i$. 

\section{Notations and definitions}\label{sec:def}
A \emph{game} is a pair ($N,v$) where $N:=\{1,\ldots,n\}$ is the set of players,
and $v:2^N\rightarrow \bR$ with $v(\emptyset)=0$. If there is no fear of
ambiguity, we will call a game simply $v$. 

We denote by 
$\mathcal{P}(N)=2^N$ the set of all subsets of $N$, while the set of
all subsets of $N$ of cardinality smaller or equal to $k$ is denoted
by $\mathcal{P}^{k}(N)$.

A game is \emph{monotone} if $S\subseteq T\subseteq N$ implies $v(S)\leq
v(T)$. It is \emph{additive} if $v(S\cup T)=v(S) + v(T)$ for every disjoint
coalitions $S,T$.  

For any game $v$, its \emph{M\"obius transform} \cite{rot64} (or Harsanyi
dividend \cite{har63}) is a set function $m:2^N\rightarrow
\bR$ defined by
\[
m(S) :=\sum_{T\subseteq S}(-1)^{|S\setminus T|}v(T), \quad \forall S\subseteq N.
\] 
If $m$ is given, it is possible to recover $v$ by $v(S) = \sum_{T\subseteq S}m(T)$.
A game $v$ is \emph{$k$-additive} \cite{gra96f} if its M\"obius transform vanishes for sets of
more that $k$ players: $m(S)=0$ if $|S|>k$, and there exists at least one
$S\subseteq N$ of $k$ players such that $m(S)\neq 0$. Note that a 1-additive
game is an additive game.

We introduce various sets:
\begin{enumerate}
\item  The set of all games with player set
$N$: $\cG(N)=\bR^{2^n-1}$;
\item The set of all monotonic games with player set
$N$: $\cMG(N)$;
\item The set of all at most $k$-additive games $\cG^k(N)=\bR^{\eta(k)}$, and at most
  $k$-additive monotonic games $\cMG^k(N)$, with
  $\eta(k):=\sum_{\ell=1}^k\binom{n}{\ell}$. Note that $\cG^1(N)$ denotes the
  set of additive games.
\item The set of \emph{selectors} on $N$:
\[
\cA(N):=\{\alpha:2^N\setminus \{\emptyset\}\rightarrow N, \quad S\mapsto\alpha(S)\in S\};
\]
\item The set of \emph{sharing functions} on $N$: 
\[
\cQ(N)=\Big\{q:2^N\setminus\{\emptyset\}\times N\rightarrow[0,1]\mid q(K,i)=0
  \text{ if } i\not\in K, \ \ \sum_{i\in K}q(K,i)=1,\quad\emptyset\neq
  K\subseteq N\Big\}.
\]
If $q$ is such that for all $K$ there exists $i\in K$ such that $q(K,i)=1$, then
$q$ is a selector. Conversely, any selector  can be viewed as a sharing
function. Moreover, $\cQ(N)$ is a convex polyhedron whose vertices are the selectors.
\end{enumerate}

Next, we introduce various mappings defined on $\cG(N)$ (or any subset like
$\cMG(N)$, $\cG^k(N)$, etc.). Most of the following notions are usually not
considered as mappings, but it is very convenient here to do so.
\begin{enumerate}
\item The \emph{preimputation set} $\PI:\cG(N)\rightarrow 2^{\bR^n}$
\[
\PI(v):=\{x\in \mathbb{R}^{N}\mid x(N)=v(N)\};
\]
\item The \emph{imputation set}  $\I:\cG(N)\rightarrow 2^{\bR^n}$
\[
\I(v):=\{x\in \mathbb{R}^{N}\mid x(N)=v(N)\ \text{and}\  x_{i}\geq
v(\{i\}) \ \text{for all}\  i\in N\};
\]
\item The \emph{core} $\core:\cG(N)\rightarrow 2^{\bR^n}$
\begin{align*}
\core(v)& :=\{x\in\bR^n\mid x(S)\geq v(S), \ \ \forall S\subset N, \text{ and }
 x(N)=v(N)\}\\ & = \{\phi\in\cG^1(N)\mid \phi(S)\geq v(S), \ \ \forall S\subset
 N,\text{ and } \phi(N)=v(N)\};
\end{align*}
\item The \emph{monotonic core} $\mcore:\cG(N)\rightarrow 2^{\bR^n}$
\[
\mcore(v):=\{\phi\in\cMG^1(N)\mid \phi(S)\geq v(S), \ \ \forall S\subset N,\text{
and } \phi(N)=v(N)\};
\]
\item The \emph{positive core} $\core_+:\cG(N)\rightarrow 2^{\bR^n_+}$
\[
\core_+(v):=\{x\in\bR^n_+\mid x(S)\geq v(S), \ \ \forall S\subset N,\text{
and } x(N)=v(N)\};
\]
\item The \emph{$k$-additive core} $\core^k:\cG(N)\rightarrow 2^{\cG^k(N)}$
\[
\core^k(v):=\{\phi\in\cG^k(N)\mid \phi(S)\geq v(S), \ \ \forall S\subset
N,\text{ and } \phi(N)=v(N)\},
\]
and similarly the $k$-additive monotonic core $\mcore^k$ and the $k$-additive
positive core $\core^k_+$;
\item The \emph{selector value} $x^\alpha:\cG(N)\rightarrow \bR^n$ for any
  selector $\alpha\in \cA(N)$
\[
x^\alpha_i(v):=\sum_{S\mid \alpha(S)=i}m^v(S), \quad i\in N
\]
where $m^v$ is the M\"obius transform of $v$;
\item The \emph{sharing value} $x^q:\cG(N)\rightarrow \bR^n$ for any sharing
function $q\in\cQ(N)$:
\[
x^q_i(v):=\sum_{S\ni i}q(S,i)m^v(S), \quad i\in N;
\] 
\item The \emph{selectope} $\S:\cG(N)\rightarrow 2^{\bR^n}$
\[
\S(v):=\conv\{x^\alpha(v)\mid \alpha\in\cA(N)\}=\{x^q(v)\mid q\in\cQ(N)\}.
\]
We may write
\[
\S=\bigcup_{q\in \cQ(N)}x^q.
\]
\item The \emph{marginal value} $p^\sigma:\cG(N)\rightarrow\bR^n$, with
  $\sigma\in\gS(N)$, the set of permutations on $N$
\[
p^\sigma_{\sigma(i)}(v):=v(S_i) - v(S_{i-1}), \quad i\in N, 
\]
where $S_i:=\{\sigma(1),\ldots,\sigma(i)\}$. Each marginal value is a selector
value (and hence a sharing value): the selector corresponding to $p^\sigma$ is
$\alpha$ which selects in $S$ the player of maximal rank\footnote{We adopt the
  convention: $i$ is the rank, and $\sigma(i)$ the player of rank $i$.}. 
\item The \emph{Weber set} $\W:\cG(N)\rightarrow 2^{\bR^n}$
\[
\W(v) = \conv\{p^\sigma(v)\mid \sigma\in\gS(N)\}.
\]
From the above remark, we have $\W(v)\subseteq \S(v)$ for any game $v$.
\item The Shapley value $\sh:\cG(N)\rightarrow \bR^n$. Since $\sh(v)\in \S(v)$,
  we write $\sh\in \S$ (particular sharing value). 
\end{enumerate}

An element $\phi$ of the $k$-additive core is a $k$-additive game. It induces by its
M\"obius transform $m^\phi$ a preimputation on all coalitions of at most $k$
players, since by definition $m^\phi(S)=0$ for all $S\subseteq N$ such that
$|S|>k$, and $\sum_{S\in\cP^k(N)}m^\phi(S)=v(N)$. We call such a (generalized)
imputation a \emph{$k$-order preimputation}. Note that in general, $m^\phi$ need
not be positive everywhere.

It remains to derive from a given $k$-order preimputation $m^\phi$ a classical
preimputation $x$, by sharing for every coalition $S\in\cP^k(N)$ the amount
$m^\phi(S)$ among players in $S$, i.e., by using a sharing function $q\in \cQ$
applied on $m^\phi$. In other words, any preimputation obtained from $m^\phi$ is
a sharing value $x^q(\phi)$ for some $q\in \cQ$, and vice-versa. It follows
that the set of preimputations derived from an element $\phi$ of the
$k$-additive core   is the selectope $\S(\phi)$. In short, the set of
preimputations which can be derived from the $k$-additive core is $\S(\core^k(v))$.

\section{Basic results and facts on the $k$-additive core}\label{sec:bare}
The $k$-additive core is a polyhedron of dimension $\sum_{i=1}^k\binom{n}{i}-1$,
possibly unbounded. A study of its vertices has been done in \cite{grmi07}, with
results similar to the Shapley-Ichiishi result for convex games. By contrast,
the monotonic $k$-additive core is always bounded, but has many more vertices
than the $k$-additive core, and it seems quite difficult to study them.

A noticeable fact shown in \cite{migr05} is that $\core^k(v)\neq\emptyset$ for
any game in $\cG(N)$, as soon as $k\geq 2$. However, this property does not hold
for the monotonic $k$-additive core. Exact conditions for nonemptiness are given
in \cite{migr05}; we call \emph{$k$-balanced-monotone} a game $v$ such that
$\mcore^k(v)\neq\emptyset$. 

We prove some elementary facts concerning the polytope $\mcore^k(v)$ and the
convex part of the $k$-additive core.
\begin{proposition}\label{prop:1}
For every $k$-balanced-monotone game $v$ on $N$, any $2\leq k\leq n$, any
sharing value $x^q$ we have:
\begin{gather*}
\ext(x^q(\mcore^k(v)))\subseteq x^q(\ext(\mcore^k(v)))\\
x^q(\mcore^k(v)) =\conv(x^q(\ext(\mcore^k(v)))).
\end{gather*}
The results holds also if $\mcore^k$ is replaced by $\conv(\ext(\core^k))$.
\end{proposition}
\begin{proof}
It is well known from the theory of polyhedra that if $P$ is a polytope in
$\bR^m$ and $\rho$ is a linear mapping from $\bR^m$ to $\bR^p$, $m\geq p$, then
$\rho(P)$ is a polytope. Moreover, a vertex $y$ of $\rho(P)$ is necessarily the
image from a vertex of $P$ (indeed, suppose that no $x$ such that $\rho(x)=y$ is a
vertex. Then it exists $x_1,x_2\in P$, $\alpha\in\left]0,1\right[$ such that
    $x=\alpha x_1+(1-\alpha)x_2$. By linearity, we get
    $y=\rho(x)=\alpha\rho(x_1)+(1-\alpha)\rho(x_2)$, contradicting the fact that
    $y$ is a vertex of $\rho(P)$), but the converse does not hold in general.

Since $\mcore^k(v)$ is a polytope and $x^q$ is linear, the first relation
holds by the above fact. Now, since $\conv(\ext(P))=P$, we get
\[
x^q(\mcore^k(v)) = \conv(\ext(x^q(\mcore^k(v))))=\conv(x^q(\ext(\mcore^k(v)))),
\]
the second equality coming from the inclusion relation.
\end{proof}
By Proposition~\ref{prop:1} and the fact that $\S=\bigcup_{q\in\cQ(N)} x^q$, we
get:
\begin{corollary}\label{cor:1}
For every $k$-balanced-monotone game $v$ on $N$ and any $2\leq k\leq n$, we have:
\[
\S(\mcore^k(v))\subseteq \conv(\S(\ext(\mcore^k(v)))).
\]
The results holds also if $\mcore^k$ is replaced by $\conv(\ext(\core^k))$.
\end{corollary}
Equality holds if $\S(\mcore^k(v))$ is a convex set, which does not seem to be
true in general.

\begin{proposition}\label{prop:3}
Suppose that $\core(v)\neq\emptyset$. Then for any $2\leq k\leq n$ and any
sharing value~$x^q$
\[
x^q(\core^k(v))\supseteq \core(v).
\]
The same result holds with the monotonic core.
\end{proposition}
\begin{proof}
Clear from the fact that $\core^k(v)\supseteq \core(v)$, and that
$\phi\in\core(v)$ (considered as an element of $\cG^1(N)$) implies $x^q(\phi)=\phi$.
\end{proof}
\begin{remark}
\begin{enumerate}
\item We have $\mcore=\core_+$. Indeed, an element of the positive core is
  monotone. Conversely, take $x$ in the monotonic core. It implies that for any
  $S\subset N$ and $i\not\in S$, we have $x(S\cup i)-x(S)=x(i)\geq
  0$. Obviously, this is no more true for the $k$-additive core. 
\item It is easy to find that $v(\{i\})\geq 0$ for all $i\in N$ is a sufficient
  condition for ensuring $\core_+=\core$ (and therefore $\mcore=\core$). But a
  similar condition for the $k$-additive case seems to be hard to find. Clearly,
  if $v^*$ dominates $v$, the monotonicity of $v$ does not imply the
  monotonicity of $v^*$ in general.
\end{enumerate}
\end{remark}

\section{Preimputations induced by the $k$-additive core}\label{sec:kadd}
We begin by recalling the following result on systems of inequalities.
\begin{lemma}\label{lem:1}
Consider the system of linear inequalities
\begin{align*}
\sum_{j=1}^n a_{ij}x_j &\leq b_i, \quad (i\in I)\\
\sum_{j=1}^n a'_{ij}x_j &= b'_i, \quad (i\in E).
\end{align*}
Consider a given $j_0\in\{1,\ldots,n\}$ such that $a'_{ij_0}=0$ for all $i\in
E$, and define $I_0:=\{i\in I\mid a_{ij_0}=0\}$ (possibly empty).  If all
$a_{ij_0}$, $i\in I\setminus I_0$, have the same sign, then the above system is
equivalent to
\begin{align*}
\sum_{j=1}^n a_{ij}x_j & \leq b_i, \quad (i\in I_0)\\
\sum_{j=1}^n a'_{ij}x_j & = b'_i, \quad (i\in E).
\end{align*}
\end{lemma}
\begin{proof}
This result may be deduced from the Fourier-Motzkin elimination. Otherwise,
simply remark the following: suppose w.l.o.g. that $a_{ij_0}> 0$ for all $i\in
I\setminus I_0$. Then, any inequality $\sum_{j=1}^n a_{ij}x_j \leq b_i, \quad
(i\in I\setminus I_0)$ will be satisfied for any
$x_1,\ldots,x_{j_0-1},x_{j_0+1},\ldots,x_n$, for a sufficiently negatively large
value of $x_{j_0}$.
\end{proof}
This observation is the key for the next theorem.
\begin{theorem}\label{th:1}
For any $x^q$, $q\in\cQ(N)$ such that $q(K,i)>0$ for all $K\subseteq N$ and
$i\in K$, for any $v\in\cG(N)$, for any $2\leq k\leq n$, we have
\[
x^q(\core^k(v)) = \PI(v).
\]
Therefore, $\S\circ\core^k=\PI$.
\end{theorem}
\begin{proof}
Take any $v\in\cG(N)$, and any $k\geq 2$. Since $\core^k(v)\neq\emptyset$, for any
$v^*\in\core^k(v)$ with M\"obius transform $m^*$, we have:
\begin{align*}
m^*(i) & \geq v(i), \quad i\in N\\
\sum_{\substack{K\subseteq S\\ |K|\leq k}}m^*(K) &\geq v(S),\quad S\subset N,
|S|>1\\
\sum_{\substack{K\subseteq N\\ |K|\leq k}}m^*(K) & =  v(N).
\end{align*}  
Take any $x^q$ with sharing system $q\in \cQ(N)$ and write for simplicity
$x:=x^q(v^*)$. We have by definition
\[
x_i = m^*(i) + \sum_{\substack{K\ni i\\2\leq |K|\leq k}}q(K,i)m^*(K).
\]
Then
\[
m^*(i) = x_i - \sum_{\substack{K\ni i\\2\leq |K|\leq k}}q(K,i)m^*(K).
\] 
Replacing in the above system, we get:
\begin{align*}
x_i - \sum_{\substack{K\ni i\\2\leq |K|\leq k}}q(K,i)m^*(K) & \geq v(i), \quad i\in N\\
\sum_{i\in S}x_i - \sum_{i\in S}\sum_{\substack{K\ni i\\2\leq|K|\leq
    k}}q(K,i)m^*(K) + \sum_{\substack{K\subseteq S\\2\leq|K|\leq k}}m^*(K) &\geq v(S),\quad S\subset N,
|S|>1\\
\sum_{i\in N}x_i - \sum_{i\in N}\sum_{\substack{K\ni i\\2\leq|K|\leq
    k}}q(K,i)m^*(K) + \sum_{\substack{K\subseteq N\\2\leq|K|\leq k}}m^*(K) &= v(N).
\end{align*}  
The second line becomes
\[
\sum_{i\in S}x_i -\sum_{\substack{K\subseteq S\\2\leq|K|\leq
  k}}m^*(K)\underbrace{\sum_{i\in K}q(K,i)}_{=1} - \sum_{\substack{K\cap S\neq
  \emptyset\\K\not\subseteq S\\2\leq|K|\leq k}}m^*(K)\sum_{i\in K\cap S}q(K,i) +
  \sum_{\substack{K\subseteq S\\2\leq|K|\leq k}}m^*(K)\geq v(S)
\]
or
\[
\sum_{i\in S}x_i - \sum_{\substack{K\cap S\neq
  \emptyset\\K\not\subseteq S\\2\leq|K|\leq k}}m^*(K)\sum_{i\in K\cap
  S}q(K,i)\geq v(S).
\]
Therefore the system becomes:
\begin{align*}
x_i - \sum_{\substack{K\ni i\\2\leq |K|\leq k}}q(K,i)m^*(K) & \geq v(i), \quad
i\in N\\
\sum_{i\in S}x_i - \sum_{\substack{K\cap S\neq
  \emptyset\\K\not\subseteq S\\2\leq|K|\leq k}}m^*(K)\sum_{i\in K\cap
  S}q(K,i) &\geq v(S),\quad S\subset N,
|S|>1\\
\sum_{i\in N}x_i & = v(N).
\end{align*}
Now, by positivity of $q$, applying Lemma~\ref{lem:1} successively to all $m^*(K)$, it remains only
\[
\sum_{i\in N}x_i  = v(N).
\]
\end{proof}
\begin{remark}
$x^q=\sh$ fulfills the condition, hence $\sh(\core^k(v))=\PI(v)$. It can
  be interpreted by saying that any preimputation can be seen as the Shapley
  value of some element (in general not unique) of the 2-additive core.
\end{remark}
We turn now to the case of selector values.
\begin{theorem}\label{th:3}
Let $\alpha\in\cA(N)$ be a selector and $x^\alpha$ the corresponding selector
value. Then, for any $2\leq k\leq n$ we have
\[
x^\alpha(\core^k(v)) = \{x\in\PI(v)\mid x(S)\geq v(S),\quad \forall S\subset N, S\not\in\cC(\alpha)\}
\]
where $\cC(\alpha):=\{S\subset N\mid \exists K\subseteq N, 2\leq |K|\leq n, K\cap
S\neq \emptyset,K\not\subseteq S,\alpha(K)\in S\}$.
\end{theorem}
\begin{proof}
From proof of Theorem~\ref{th:1}, we know that, for any $v^*\in\core^k(v)$ with
M\"obius transform $m^*$, the system of inequalities is:
\begin{align*}
x_i - \sum_{\substack{K\ni i\\2\leq |K|\leq k}}q(K,i)m^*(K) & \geq v(i), \quad
i\in N\\
\sum_{i\in S}x_i - \sum_{\substack{K\cap S\neq
  \emptyset\\K\not\subseteq S\\2\leq|K|\leq k}}m^*(K)\sum_{i\in K\cap
  S}q(K,i) &\geq v(S),\quad S\subset N,
|S|>1\\
\sum_{i\in N}x_i & = v(N),
\end{align*}
with $q$ corresponding to $x^\alpha$, that is, $q(K,i)=1$ if and only if
$\alpha(K)=i$, and 0 otherwise. Therefore, the system becomes:
\begin{align*}
x_i - \sum_{\substack{K\mid \alpha(K)=i\\2\leq |K|\leq k}}m^*(K) & \geq v(i), \quad
i\in N \text{ s.t. } \alpha(K)=i\text{ for some } K, 2\leq|K|\leq k\\
x_i & \geq v(i), \quad i\in N \text{ otherwise}\\
\sum_{i\in S}x_i - \sum_{\substack{K\cap S\neq
  \emptyset\\K\not\subseteq S\\ \alpha(K)\in S\\2\leq|K|\leq k}}m^*(K) &\geq v(S),\quad S\subset N,
|S|>1\text{ s.t. } S\in\cC(\alpha)\\
\sum_{i\in S}x_i & \geq v(S),\quad S\subset N,
|S|>1\text{ s.t. } S\not\in\cC(\alpha)\\
\sum_{i\in N}x_i & = v(N).
\end{align*}
Applying Lemma~\ref{lem:1} successively to all $m^*(K)$, we get the result. 
\end{proof}
Therefore, $x^\alpha(\core^k(v))$ is a superset of $\core(v)$, whose structure
depends on which coalitions appear in the inequalities. We can be more specific
by taking marginal values, which are particular selector values.
\begin{theorem}\label{th:2}
For any permutation $\sigma\in\gS$, for any $v\in\cG(N)$, for any $2\leq k\leq
n$, we have
\begin{equation}\label{eq:psig}
p^\sigma(\core^k(v)) = \Big\{x\in\PI(v)\mid \sum_{j=1}^i x_{\sigma(j)}\geq
v(\{\sigma(1),\ldots,\sigma(i)\}),\quad i=1,\ldots,n-1\Big\}.
\end{equation}
\end{theorem}
\begin{proof}
This time $q$ corresponds to $p^\sigma$, that is:
\[
q(K,i)=1 \text{ if and only if } \ell_\sigma(K)=i 
\]
and $q(K,i)=0$ otherwise, where $\ell_\sigma(K)$ is the last element of $K$ in
the order $\sigma$. Therefore, proceeding as for Theorem~\ref{th:3}, we get:
\begin{align*}
x_{\sigma(1)} & \geq v(\{\sigma(1)\})\\
x_i - \sum_{\substack{K\mid \ell_\sigma(K)=i\\2\leq |K|\leq k}}m^*(K) &
\geq v(i), \quad i=2,\ldots,n\\ 
x_{\sigma(1)}+\cdots +x_{\sigma(i)} & \geq
    v(\{\sigma(1),\ldots,\sigma(i)\}),\quad i=2,\ldots,n-1  \\
\sum_{i\in S}x_i - \sum_{\substack{K\cap S\neq
  \emptyset\\K\not\subseteq S\\ \ell_\sigma(K)\in S\\2\leq|K|\leq k}}m^*(K)&\geq v(S),\quad S\subset N,
S\neq\{\sigma(1),\ldots,\sigma(i)\}\text{ for some }i\in N\\
\sum_{i\in N}x_i & = v(N).
\end{align*}
(note that in the 4th inequality, the set of $K$ such that $K\cap S\neq
  \emptyset,K\not\subseteq S, \ell_\sigma(K)\in S,2\leq|K|\leq k$ is never empty due
to the assumption that $S$ is not one of the sets in the chain induced by $\sigma$)
Applying Lemma~\ref{lem:1} successively to all $m^*(K)$, we get the result. 
\end{proof}

Let us study the structure of $p^\sigma(\core^k(v))$. We denote by
$C_\sigma$ the maximal chain associated to $\sigma$, that is,
$C_\sigma=\{S_1,\ldots, S_n\}$, with $S_i:=\{\sigma(1),\ldots,\sigma(i)\}$. We
introduce $v_{|C_\sigma}$ the restriction of $v$ to the maximal chain
$C_\sigma$ (game with restricted cooperation). Therefore, we have
\[
p^\sigma(\core^k(v)) =\core(v_{|C_\sigma}).
\]
The core of games with restricted cooperation has been largely studied (see a
survey in \cite{gra09a}). We can derive directly from known results the
structure of $p^\sigma(\core^k(v))$. A first observation is that this set is
nonempty, for it contains  the marginal value $p^\sigma(v)$, obviously solution
of the set of inequalities (\ref{eq:psig}). Moreover, $p^\sigma(v)$ is the
unique vertex of $p^\sigma(\core^k(v))$. Indeed, there are $n-1$ inequalities and
one equality in (\ref{eq:psig}), hence from it we can derive only one set of $n$
equalities, which precisely defines $p^\sigma(v)$. 

It remains to find the extremal rays of $p^\sigma(\core^k(v))$, which can be
obtained by known results about the core of games with restricted
cooperation. The collection $C_\sigma$ being a chain, it is a distributive
lattice of height $n$, whose
join-irreducible elements are  $\{\sigma(1)\}$,
$\{\sigma(1),\sigma(2)\}$,\ldots,$\{\sigma(1),\ldots,\sigma(n)\}$. We cite the
following result due to Tomizawa.
\begin{proposition}
(Tomizawa \cite{tom83}, also cited in Fujishige \cite[Th. 3.26]{fuj05}) Let
  $\cF\subseteq 2^N$ be a distributive lattice of height $n$, with joint
  irreducible elements $J_1,\ldots, J_n$. The
  extremal rays of $\core(0)$, the recession cone of $\core(v)$, are  of the
  form $(1_j,-1_i)$, with $i\in N$ such that the smallest join-irreducible
  element $J$ containing $i$ satisfies $|J|>1$, and $j\in J$, such that the
  smallest join-irreducible element $J'$ containing $j$ is the predecessor of
  $J$ in the poset of joint-irreducible elements. 
\end{proposition}
Applying this result to our case, we find the following. 
\begin{proposition}
For any permutation $\sigma\in\gS$, for any $v\in\cG(N)$, for any $2\leq k\leq
n$, $p^\sigma(\core^k(v))$ is a pointed unbounded polyhedron, with unique vertex
$p^\sigma(v)$, and extreme rays given by $\1_{\sigma(i-1)}-\1_{\sigma(i)}$,
$i=2,\ldots,n$. 
\end{proposition} 
As a consequence of Theorem~\ref{th:3}, Theorem~\ref{th:2} and Proposition~\ref{prop:3}, we have immediately
\begin{theorem}\label{th:4}
For any $2\leq k\leq n$, for any $v\in \cG(N)$, we have
\[
\bigcap_{\sigma\in\gS}p^\sigma(\core^k(v))=\bigcap_{\alpha\in\cA(N)}x^\alpha(\core^k(v)) = \bigcap_{x\in \S}x(\core^k(v)) =
\core(v).
\]
\end{theorem}

\section{Preimputations induced by the monotonic $k$-additive core}\label{sec:kadd1}
The case of the monotonic core is much more tricky to study, because
monotonicity induces supplementary inequalities which make Lemma~\ref{lem:1}
inapplicable. Nevertheless, a result can be derived for the case of selector
values.

A first simple but important observation is the following.
\begin{lemma}\label{lem:2}
If $P_1,P_2$ are two polyhedra in $\bR^m$, and $\rho$ is a linear mapping from
$\bR^m$ to $\bR^p$, then $\rho(P_1\cap P_2)=\rho(P_1)\cap \rho(P_2)$.
\end{lemma}
\begin{proof}
Let $P_1,P_2$ be defined by sets of $k$ and $\ell$ linear inequalities
respectively, in variables $x_1,\ldots, x_m$. Then $P_1\cap P_2$ is defined by
the union of these two sets of inequalities. Now, $\rho(P_1),\rho(P_2)$ are
polyhedra defined by sets of $k$ and $\ell$ linear inequalities in variables
$y_1,\ldots, y_p$, and the union of these two systems defines $\rho(P_1)\cap
\rho(P_2)$. But this system is also the transform by $\rho$ of the system
defining $P_1\cap P_2$, hence it represents $\rho(P_1\cap P_2)$ as well.  
\end{proof}
We apply this result with $P_1=\core^k(v)$ and $P_2$ the polyhedron of monotone
$k$-additive games. Clearly $\mcore^k(v)=P_1\cap P_2$, and it suffices to study
the transform of $P_2$ by sharing values. Polyhedron $P_2$ reads, in the space
of the M\"obius transform (see \cite{chja89}):
\[
P_2=\big\{m\in\bR^{\eta(k)}\mid \sum_{\substack{K\ni i, K\subseteq
    S\\ 1\leq|K|\leq k}}m(K)\geq 0,\quad S\subseteq N, S\neq \emptyset,i\in S\big\},
\]  
with $\eta(k):=\sum_{\ell=1}^k\binom{n}{\ell}$.
\begin{theorem}\label{th:5}
Let $\alpha\in\cA(N)$ be a selector satisfying the following property: if $K,K'$
are such that $2\leq |K|,|K'|\leq k$ and $q(K,i)=q(K',i)=1$ for some $i\in N$,
then any $K''\subseteq K\cup K'$ such that $2\leq |K''|\leq k$ satisfies
$q(K'',i)=1$. Let $x^\alpha$ be the corresponding selector
value. Then, for any $2\leq k\leq n$, for any $k$-balanced-monotone game $v$ we have
\[
x^\alpha(\mcore^k(v)) = \{x\in\PI(v)\cap \bR^n_+\mid x(S)\geq v(S),\quad \forall S\subset N,
S\not\in\cC(\alpha)\}
\]
where $\cC(\alpha):=\{S\subset N\mid \exists K\subseteq N, 2\leq |K|\leq n, K\cap
S\neq \emptyset,K\not\subseteq S,\alpha(K)\in S\}$.
\end{theorem}
\begin{proof}
From Lemma~\ref{lem:2}, it suffices to compute $p^\sigma(P_2)$. For ease, we
express $x^q(P_2)$ for any sharing value $q\in\cQ(N)$. For any game $v$ in $P_2$
with M\"obius transform $m$ we have
\begin{align*}
m(i) & \geq 0,\quad i\in N\\
\sum_{\substack{K\ni i,K\subseteq S\\1\leq |K|\leq k}}m(K) & \geq 0,\quad
S\subseteq N, |S|>1, i\in S.
\end{align*}
Writing for simplicity $x:=x^q(v)$, we have by definition
\[
x_i = m(i) + \sum_{\substack{K\ni i\\2\leq |K|\leq k}}q(K,i)m(K),
\]
so that
\[
m(i) = x_i-\sum_{\substack{K\ni i\\2\leq |K|\leq k}}q(K,i)m(K),
\]
for all $i\in N$.  Replacing in the above system, we get:
\begin{align}
x_i-\sum_{\substack{K\ni i\\2\leq |K|\leq k}}q(K,i)m(K) & \geq 0,\quad i\in N\label{eq:mc1}\\
x_i + \sum_{\substack{K\ni i\\K\subseteq S\\2\leq|K|\leq k}}(1-q(K,i))m(K) -
\sum_{\substack{K\ni i\\K\not\subseteq S\\2\leq|K|\leq k}}q(K,i)m(K) & \geq
0,\quad S\subseteq N,|S|>1,i\in S.\label{eq:mc2}
\end{align}
Let us consider the selector value $x^\alpha$ for some $\alpha\in\cA(N)$. Then
$q(K,i)=1$ or 0, for every $K\subseteq N$, $2\leq|K|\leq k$, and every $i\in
N$. Consider $i$ to be fixed, and denote by $\cK$ the collection of sets $K$,
$2\leq |K|\leq k$, such that $q(K,i)=1$. If $\cK$ is empty, then (\ref{eq:mc1})
reduces to
\[
x_i\geq 0.
\]
Suppose then that $\cK$ is not the empty collection, and consider $S=\bigcup
\cK$, i.e., the union of all sets in $\cK$. Observe that $|S|>1$ and $S\ni
i$. Equation (\ref{eq:mc2}) for this $S$ gives:
\[
x_i\geq 0,
\]
since by hypothesis any $K\subseteq S$ such that $K\ni i$ and $2\leq |K|\leq k$
is a member of $\cK$, and any $K\not\subseteq S$ satisfies $q(K,i)=0$ by
definition of $\cK$. Hence in any case, we get the inequality $x_i\geq 0$. This
reasoning can be done for any $i\in N$, therefore $x_1,\ldots, x_n\geq 0$.

The remaining inequalities have the form
\[
x_i+\sum_{\substack{K\ni i\\2\leq |K|\leq k}}\epsilon m(K)\geq 0,
\]
with $\epsilon=0$, 1 or $-1$.
Let us eliminate all variables $m(K)$ by Fourier-Motzkin elimination from this
system of inequalities. Observe that elimination amounts to add pairs of
inequalities, possibly multiplied by some positive constants. Therefore, as a
result of elimination, it will remain only inequalities of the form
\[
a_1x_1+\ldots+a_nx_n\geq 0,
\]
with $a_1,\ldots, a_n\geq 0$, so that they are all redundant with
$x_1,\ldots,x_n\geq 0$. 
\end{proof}
Observe that any marginal value $p^\sigma$ is a sharing value having the
property requested in Theorem~\ref{th:5}. Therefore, we have the following result.
\begin{corollary}\label{cor:2}
For any permutation $\sigma\in\gS$, for any $v\in\cG(N)$, for any $2\leq k\leq
n$, for any $k$-balanced-monotone game $v$, we have
\begin{equation}\label{eq:psigm}
p^\sigma(\mcore^k(v)) = \Big\{x\in\PI(v)\cap \bR^n_+\mid \sum_{j=1}^i x_{\sigma(j)}\geq
v(\{\sigma(1),\ldots,\sigma(i)\}),\quad i=1,\ldots,n-1,\Big\}.
\end{equation}
\end{corollary}

\bibliographystyle{plain}
\bibliography{../BIB/fuzzy,../BIB/grabisch,../BIB/general}

\end{document}